\documentclass[a4paper,USenglish,10pt]{article}
\usepackage[utf8]{inputenc}

\usepackage{amsthm,amssymb,amsfonts,amsmath}
\newtheorem{theorem}{Theorem}
\newtheorem{lemma}[theorem]{Lemma}

\newtheorem{observation}[theorem]{Observation}
\theoremstyle{definition}
\newtheorem{definition}[theorem]{Definition}

\usepackage[hidelinks]{hyperref}
\usepackage[blocks]{authblk}

\usepackage{graphicx}
\graphicspath{{figures/}}
\usepackage{wrapfig}

\newcommand{\fig}[1]{\figurename~\ref{#1}}

\makeatletter
\g@addto@macro\bfseries{\boldmath}
\makeatother

\advance\hoffset-7mm
\begin{document}

\title{A Note on the Flip Distance Problem for Edge-Labeled Triangulations}

\author{Alexander Pilz\thanks{Supported by a Schr\"odinger fellowship of the Austrian Science Fund (FWF): J-3847-N35.}}
\affil{Institute of Software Technology, Graz University of Technology.
  \texttt{apilz@ist.tugraz.at}}

\maketitle

\begin{abstract}
For both triangulations of point sets and simple polygons, it is known that determining the flip distance between two triangulations is an NP-hard problem.
To gain more insight into flips of triangulations and to characterize ``where edges go'' when flipping from one triangulation to another, flips in edge-labeled triangulations have lately attracted considerable interest.
In a recent breakthrough, Lubiw, Mas\'arov\'a, and Wagner (in Proc.\ 33rd Symp.\ of Computational Geometry, 2017) prove the so-called ``Orbit Conjecture'' for edge-labeled triangulations and ask for the complexity of the flip distance problem in the edge-labeled setting.
By revisiting and modifying the hardness proofs for the unlabeled setting, we show in this note that the flip distance problem is APX-hard for edge-labeled triangulations of point sets and NP-hard for triangulations of simple polygons.
The main technical challenge is to show that this remains true even if the source and target triangulation are the same when disregarding the labeling.
\end{abstract}

\section{Introduction}
A \emph{triangulation of a point set} in the plane is a maximal straight-line crossing-free graph whose vertex set is exactly the point set.
Similarly, a \emph{triangulation of a simple polygon} is a maximal straight-line crossing-free graph whose edges are boundary edges and diagonals of the polygon and whose vertex set is exactly that of the polygon.
For two triangles of the triangulation forming a convex quadrilateral, the operation of removing their common edge and replacing it by the other diagonal of the quadrilateral is called a \emph{flip}.
The result of a flip is again a triangulation.
The \emph{flip graph} of a given polygon or point set is the graph whose vertex set is the set of all triangulations, and in which two triangulations share an edge iff one can be obtained from the other by a flip.
For both triangulations and polygons, it is known that one can transform each triangulation into any other~\cite{lawson_connected,lawson_delaunay} using $O(n^2)$ flips, i.e., the flip graph is connected with diameter $O(n^2)$.
For triangulations on $n$ vertices, its diameter can be $\Omega(n^2)$, both for polygons and point sets~\cite{hurtado_noy_urrutia}.
The \emph{flip distance} of two triangulations is the distance of the corresponding vertices in the flip graph.

Determining the flip distance between two triangulations is APX-hard for point sets~\cite{point_set_hard} (see also~\cite{lubiw_pathak}), and NP-hard for simple polygons~\cite{poly_hard}.
The problem variant for convex polygons remains a long-standing open problem.

In order to improve the understanding of flip graphs, Bose et al.~\cite{edge_labelled} considered the problem variant in which each edge of the triangulation gets a unique label, and when flipping an edge, the added edge gets the label of the removed edge.
In that setting, two labeled triangulations are equivalent if they have the same triangulation and also the labels match.
The flip graph of edge-labeled triangulations can be defined as before.
Note that it now may contain several vertices corresponding to edge-labeled triangulations with the same unlabeled triangulation.
However, the flip graph may no longer be connected~\cite{edge_labelled}.

Eppstein~\cite{eppstein} considered the flip distance problem of triangulations on point sets without convex 5-holes, i.e., without convex pentagons spanned by the points that do not contain a point of the set in their interior.
(It is known that such sets of at least 10 points need to contain collinear point triples~\cite{harborth}; see also~\cite{empty5gon}.)
As pointed out in~\cite{orbit_conjecture_socg}, Eppstein showed that no permutations of edge labels are possible via flips.
He obtains a polynomial-time algorithm to determine the flip distance of any two triangulations of such a point set.
This encourages the investigation of edge-labeled triangulations in relation to the flip distance problem.
Convex pentagons indeed play a special role for edge-labeled triangulations.

\begin{observation}[Bose et al.~\cite{edge_labelled}]
In a triangulation of a convex pentagon, the labels of its diagonals can be swapped by a sequence of five flips.
\end{observation}

Since the flip graph of edge-labeled triangulations may not even be connected, a natural computational problem is the one of determining whether two triangulations are connected.
Lubiw, Mas\'arov\'a, and Wagner~\cite{orbit_conjecture_socg} provide a polynomial-time algorithm for that problem by identifying a fundamental property of flip graphs of edge-labeled triangulations, called the \emph{orbit theorem} (see below).
They ask whether the problem of determining the flip distance of two triangulations becomes tractable in the edge-labeled setting.
It turns out that the two reductions of~\cite{point_set_hard} and~\cite{poly_hard} can be modified to answer the question in the negative, which we present in this note.

\section{Preliminaries}
This note is based on modifying the constructions in~\cite{point_set_hard} and \cite{poly_hard}.
We borrow definitions from those papers, and will refer to them for technical details of the construction.
The underlying vertex set will mainly be the same as for those reductions, we merely modify one gadget, alter the source and target triangulation (eventually making them identical when disregarding the labeling) and add the edge labels appropriately.

We assume all point sets to be in general position, i.e., no three points are collinear.
An \emph{edge of a point set $S$} is a line segment spanned by two points of~$S$.
Two edges $e$ and $f$ of~$S$ are in the same \emph{orbit} if there exists a triangulation containing $e$ that can be flipped to a triangulation that contains $f$ such that $f$ has the same label as~$e$.
As it turns out, the orbits also determine whether one labeled triangulation can be transformed into another.

\begin{theorem}[Orbit theorem~\cite{orbit_conjecture_socg}]
Given two edge-labeled triangulations of a point set, there exists a flip sequence that transforms one into the other if and only if all edges sharing a label belong to the same orbit.
\end{theorem}

While the orbit theorem is proved for point sets, its proof also directly generalizes to point sets with fixed edges, and thus to triangulations of polygons.%
\footnote{Personal communication with Zuzana Mas\'arov\'a, 2018.}

We make use of the quadrilateral graph, defined by Eppstein~\cite{eppstein}.
The \emph{quadrilateral graph} of a point set $S$ is the graph whose vertices consist of the $\binom{n}{2}$ segments spanned by two points of $S$ (i.e., \emph{edges} of the point set), and there is an edge between two segments if and only if they cross and the convex quadrilateral formed by their endpoints is empty of other points of~$S$.
The quadrilateral graph can be defined analogously for polygonal domains.
Each flippable edge in a triangulation defines a quadrilateral.
Clearly, for two edges to be in the same orbit, it is necessary that they are connected in the quadrilateral graph.
By keeping in mind that the flip graph is connected even with some edges fixed, it can also be seen that this condition is sufficient.

Our reductions will mainly make use of the following fact.
For one label to be transferred from one edge to another, the sequence of triangulations must define a path in the quadrilateral graph that connects the two edges.
Indeed, any flip sequence defines, for a given edge, a walk in the quadrilateral graph in a natural way.

A main ingredient for both reductions is the so-called \emph{double chain}, used by Hurtado, Noy, and Urrutia~\cite{hurtado_noy_urrutia} to show the quadratic lower bound on the flip graph diameter.
We will use double chains for point sets as well as for simple polygons (actually, a small double chain will be sufficient).
We thus follow the formulation of~\cite{pilz_thesis}, in which an attempt for a sufficiently general definition is made.

\begin{definition}
A \emph{double chain} $D$ is a point set of $2n$ points, $n$ on the upper chain and $n$ on the lower chain.
Let these points be $(u_1,\dots,u_n)$ and $(l_1,\dots,l_n)$, respectively, ordered from left to right.
Any point on one chain sees every point of $D$ on the convex hull boundary of the other chain (i.e., the interior of the straight line segment between these two points does not intersect the convex hulls of the two chains), and any quadrilateral formed by three points of one chain and one point of the other chain is non-convex.
See \fig{fig:chain}.
\end{definition}

Note that this definition requires the upper and the lower chain to have the same number of points.
It is also common to use generalizations where these number differ (and we will actually use one with two points on the upper and three points on the lower chain).

\begin{figure}
\centering
\includegraphics{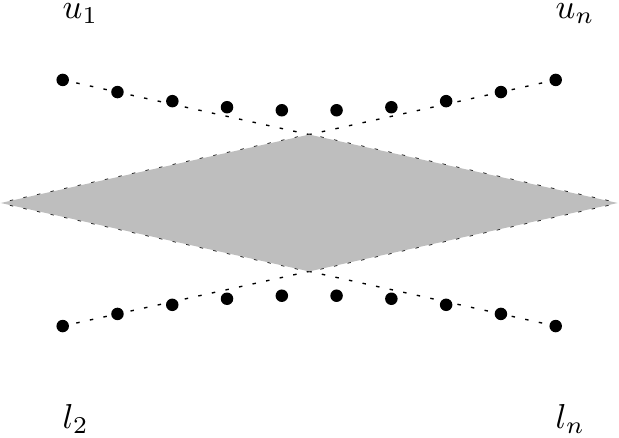}
\caption{A double chain.
The flip-kernel is shown gray.}
\label{fig:chain}
\end{figure}

\begin{definition}
Let $D$ be a double chain of $2n$ points, and consider the convex hulls of the upper and the lower chain.
The \emph{flip-kernel} of a double chain $D$ (shown gray in \fig{fig:chain}) is the maximal subset of the plane such that, for all $i,j$ with $1 \leq i,j \leq n$, and every point $p$ in the flip-kernel, the segments $p u_i$ and $p l_j$ are both interior-disjoint with the convex hulls of the upper and lower chain.
\end{definition}

Consider a triangulation of a double chain.
A crucial observation of Hurtado, Noy, and Urrutia~\cite{hurtado_noy_urrutia} for their lower bound proof is that a line separating the two chains passes through the edges from upper to lower points in a certain order, and flipping an edge $u_i l_j$ between the upper and the lower chain does not change that order.
We call such a line a \emph{separating line} of the double chain.
Under the aspect of labeled triangulations, this means that the order of the labels in which they appear along that separating line remains the same:
these edges do not have an empty convex pentagon in their orbit and thus cannot change their label with an edge of the same triangulation.

The previous reductions use this in the following way:
in order to change the triangulation of a double chain (of which multiple instances will be used in gadgets of the reduction), these edges need to be flipped to be incident to points of the set in the flip-kernel.
This process will force triangulations of certain types to appear in the flip sequence.

We call a triangulation of a double chain in which all points of the lower chain are incident to $u_1$ the \emph{upper extreme triangulation}, and the one in which all edges are incident to $l_1$ the \emph{lower extreme triangulation}.

Flipping from an upper extreme triangulation to a lower extreme triangulation requires at least $(n-1)^2$ flips~\cite{hurtado_noy_urrutia}.
If there is one additional point in the flip-kernel, this requires only $4n-4$ flips~\cite{problemas}.
Note that also by this way of flipping, the order of the labels in which a separating line crosses the edges between points of the chain does not change.

It is instructive to see how to change this order of labels for double chains with additional points.
Figure~\ref{fig:swapping_labels} shows how to change the labels of the lower extreme triangulation using an additional point in the flip-kernel of the double chain.
There, we see that it requires at most $9n-9$ flips to exchange the labels of the edges incident to $l_1$ with the labels of those incident to $u_n$, maintaining their relative order.
Apart from that, the order of the labels along a separating line stays the same.
If the point set would only consist of the double chain (and, in particular, no point in the flip-kernel), it would not be possible to change the labels in that way.

\begin{figure}
\centering
\includegraphics{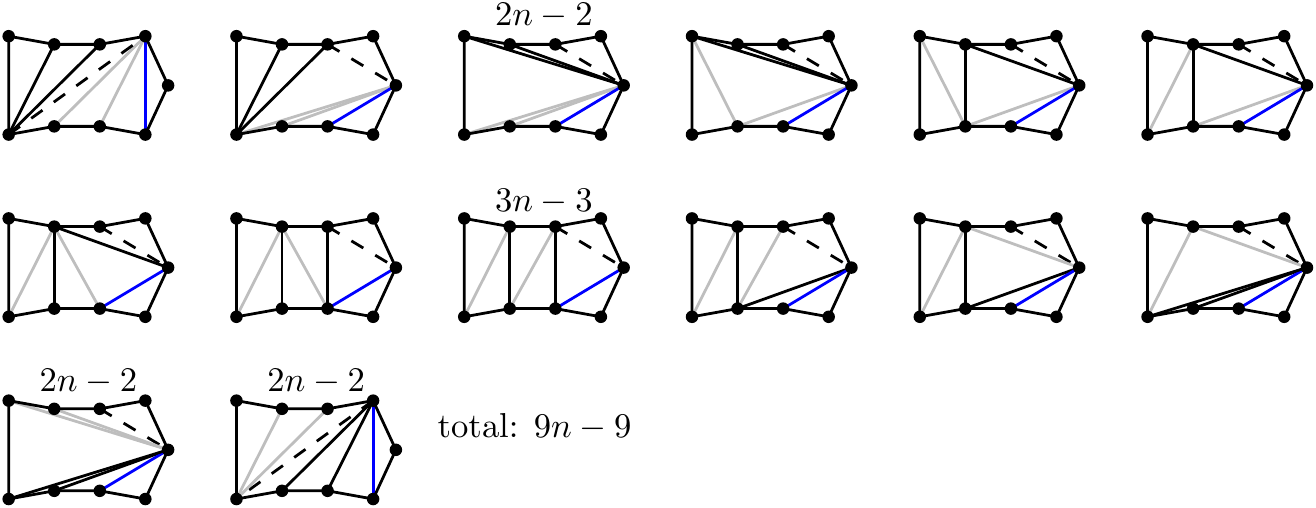}
\caption{Relabeling of a lower extreme triangulation of an extended double chain.
The colors indicate the classes of labels.
First, all edges are made incident to the point $p$ in the flip-kernel (requiring $2n-2$ flips).
Then, the edges of the two different label classes (gray and black) are arranged in a zig-zag fashion ($3n-3$ flips).
After that, the edges are again made incident to the point~$p$ ($2n-2$ flips).
Finally, using $2n-2$ flips, we flip the edges back to a lower extreme triangulation, which has the two groups of labels exchanged.
In total, this requires $9n-9$ flips.
}
\label{fig:swapping_labels}
\end{figure}

Both previous reductions rely on going from an upper to a lower extreme triangulation without using a quadratic number of flips in any double chain gadget.
One way of re-using the reduction for the stronger setting in which the source and the target have the same (unlabeled) triangulation, but have different labels, would be to change the labels of the double chain gadgets.
However, it is hard to argue in a local way that the labels ``stay within the gadgets''.
We will therefore use a modification that actually is less involved:
it is sufficient to use a construction similar to a double chain, that uses only five vertices.

\section{Point Sets}
We modify the reduction for the unlabeled variant described in~\cite{point_set_hard}.
The reduction for point sets is from the \textsc{Minimum Vertex Cover} problem, which is known to be APX-complete~\cite{vertex_cover_apx}.
We create a triangulation $T$ on a point set $S$, and argue that it is APX-hard to find the shortest sequence of flips that produces a given permutation of the edge labels.

Let $G=(V,E)$ be a graph of a \textsc{Minimum Vertex Cover} instance.
The vertices $V$ are placed as points in convex position, connected by straight-line edges representing the edges of $E$, in a way that no three edges intersect in a common point of their relative interior.%
\footnote{The appendix of~\cite{point_set_hard} contains a detailed account on accurate embedding of the points with coordinates polynomial in the size of the initial problem instance.
We will detail the precise point placement only when necessary, and otherwise refer only to the techniques used there.}
An embedding of such a graph is shown in \fig{fig_k33}.
We now add our gadgets, which consist of point sets together with edges of the triangulation.

\begin{figure}
\centering
\includegraphics{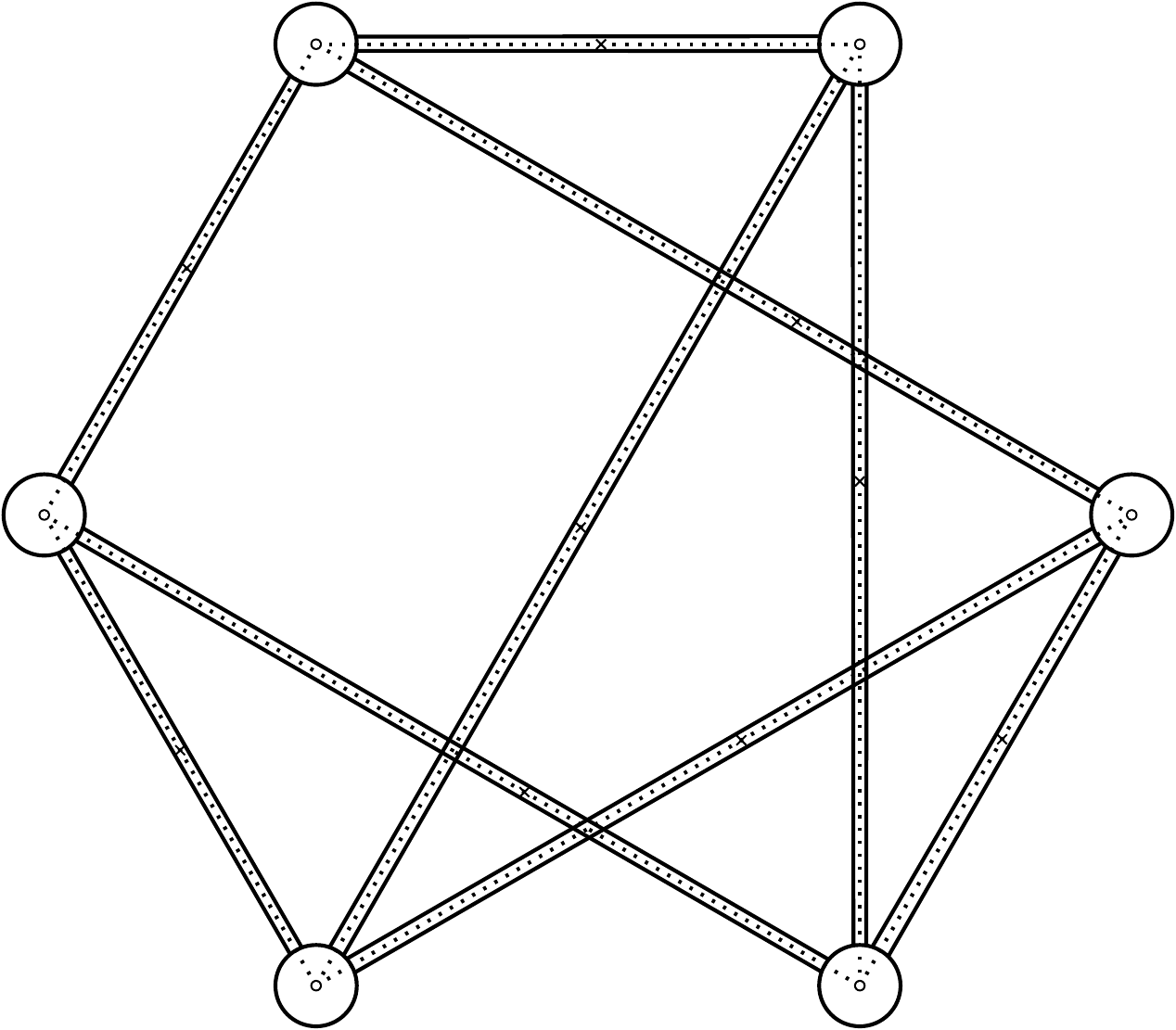}
\caption{A graph from a \textsc{Minimum Vertex Cover} instance drawn for the reduction.
The initial vertices are in convex position, and we place additional points for wiring gadgets in a small circle around it.
Edges are replaced by tunnels in which we will place points for two crossing edges and also the edge center gadget, which are the only parts that contain edges with different labels.
Figure from~\cite{point_set_hard}.
}
\label{fig_k33}
\end{figure}

The vertices of $G$ are replaced by \emph{wiring} gadgets.
These are gadgets that contain the initial point $v$ representing the vertex, called the wiring's \emph{defining vertex}.
The remaining points are placed on a small circle with center $v$ in the following way.
We place two chains $L$ and $R$ of $w-1$ points each on the circle (the value of $w$ is large but polynomial in $n$) such that any line defined by one point of $L$ and one point of $R$ separates $v$ from the remaining construction.
In the triangulation, there is a zig-zag path containing all these points and alternating between $L$ and $R$.
See \fig{fig_wiring}.
Intuitively, the edges of a wiring being flipped as shown in \fig{fig_wiring}~(right) in some triangulation of the sequence corresponds to it being chosen for the vertex cover (this will be detailed later).

\begin{figure}
\centering
\includegraphics[width=1\textwidth]{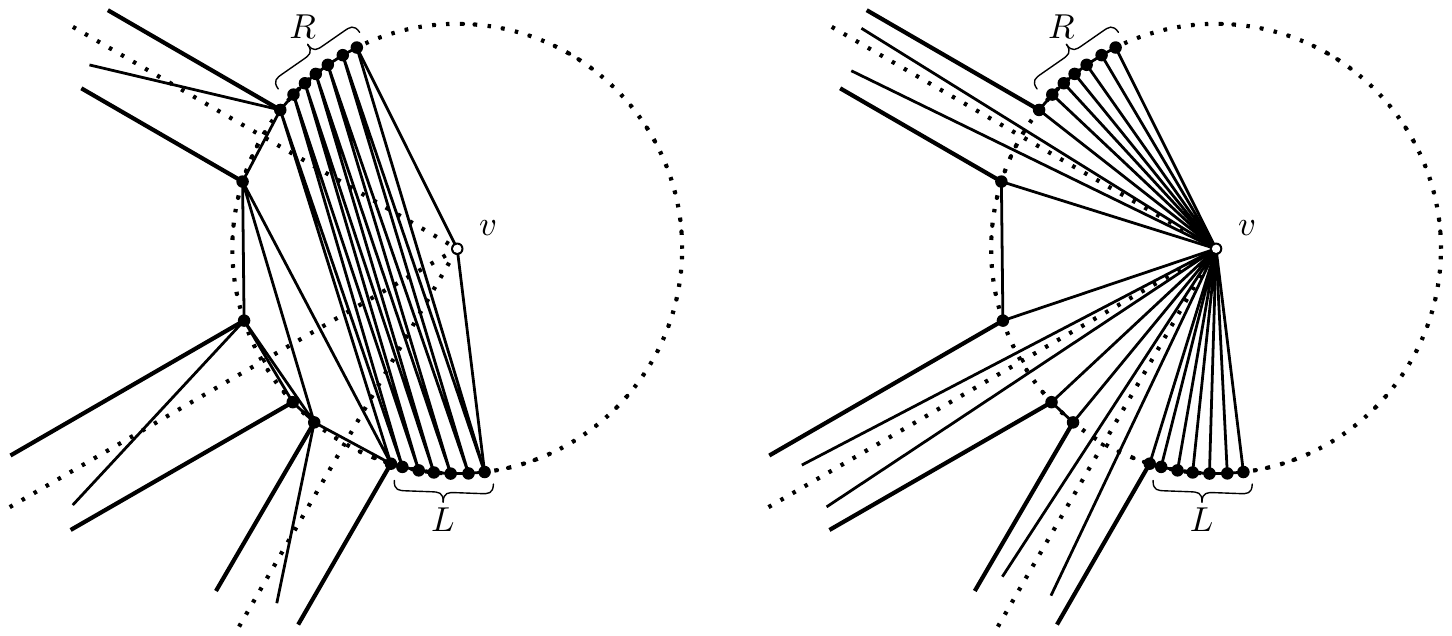}
\caption{A wiring gadget, placed at each vertex of the initial graph.
For the labels in the edge centers to be exchanged, they need to ``see'' the vertex in the middle of the wiring.
To this end, we need $2w-1$ flips.
Figure adapted from~\cite{point_set_hard}
}
\label{fig_wiring}
\end{figure}

The edges of $G$ are transformed using two types of gadgets, the \emph{crossings} and the \emph{edge centers}.
Consider an edge $vw$.
In the original reduction, the edge center consists of a double chain with exactly the points $v$ and $w$ in their flip-kernel.
In our variant of the reduction, we only have to place five points;
on the one hand, fewer points are needed, as we no longer have to rely on the number of flips to change the triangulation of a double chain to enforce certain sequences of flips, on the other hand these points need to be carefully placed to lead to the desired behavior.

For each edge $e = vw$ we arbitrarily choose a point $c_e$ that is not on a crossing (this point will not be part of the point set $S$).
Let $\vec{t}$ be a sufficiently short vector perpendicular to $e$, and translate two copies of $e$ by $\vec{t}$ and $-\vec{t}$.
We obtain two non-intersecting circular arcs $A_e$ and $A'_e$ by bending these copies slightly towards the midpoint of~$e$.
(Intuitively, we obtain a region shaped like a biconcave lens that contains the edge~$e$.)

For each crossing of two edges $e$ and $f$ of $G$, we place a \emph{crossing} gadget.
It consists of four points that are at the crossings of the arcs $A_e, A'_e, A_f, A'_f$ (in an actual embedding of~$S$ with rational coordinates, these points may only be sufficiently close to these crossings, see~\cite{point_set_hard}).
We add the edges of a triangulation of these four points to $T$.

The \emph{edge center} gadget is different to the one in~\cite{point_set_hard}.
Consider an edge~$e$ of $G$.
In the vicinity of the point $c_e$ between the arcs $A_e$ and $A'_e$, we place two points $t_1$ and $t_2$ such that the line $t_1 t_2$ is parallel to $e$.
In the same way, we place two points $b_1$ and $b_2$ on the other side of $e$, but move each of them slightly away from $c_e$.
We then place a fifth point $c$ on the midpoint of the edge $b_1 b_2$ and move it slightly towards the edge $e$ such that it is in the interior of the convex hull of $\{t_1, t_2, b_1, b_2\}$, but ``arbitrarily close'' to the edge $b_1 b_2$ (i.e., moving the point closer to the edge $b_1 b_2$ would not change the set of triangulations on the point set).
See \fig{fig:point_set_clause}.
The edges $ct_1$ and $ct_2$ are part of the triangulation~$T$, with their labels swapped in the source and target triangulation.
All other edges of the triangulation do not change their label.

\begin{figure}
\centering
\includegraphics{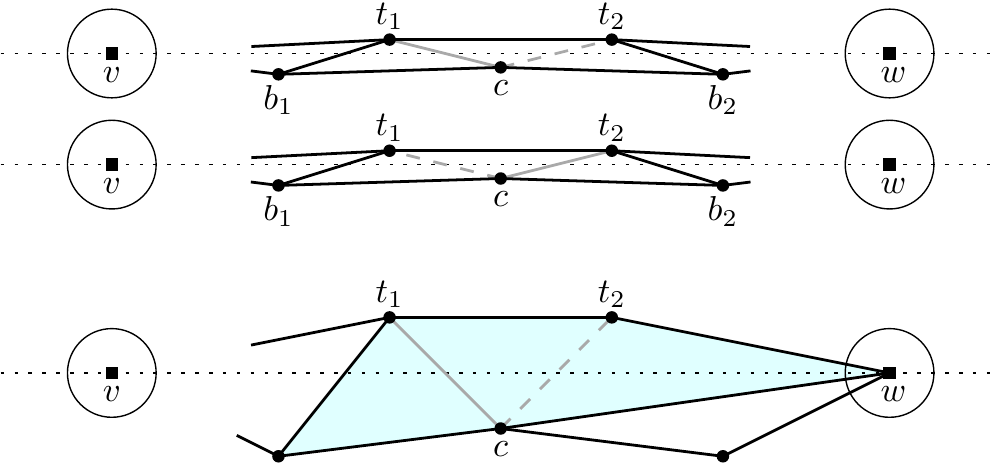}
\caption{The flat double chain of a edge center gadget.
The gray edges can only be flipped inside the flat double chain or when being made incident to one of the two central vertices $v$ or $w$ of the corresponding variable gadgets:
The component in the quadrilateral graph that contains that edge can be disconnected by removing the quadrilaterals (i.e., edges in the quadrilateral graph) containing the central vertices of the variable gadget;
therefore, the wiring of at least one corresponding variable gadget needs to be flipped.
A vertically scaled figure is shown at the bottom, with the blue region indicating a convex pentagon.
The vertices $v$ and $w$ should be considered far away, and the five vertices are sufficiently close to the edge $vw$ such that any further quadrilateral spanned with $cb_1$ or $cb_2$ contains $t_1$ and $t_2$, respectively, and any further quadrilateral spanned with $t_1t_2$ contains~$c$.
In general, there will be vertices of the crossing gadgets between the edge center and the wiring gadget.
}
\label{fig:point_set_clause}
\end{figure}

\begin{figure}
\centering
\includegraphics[page=2]{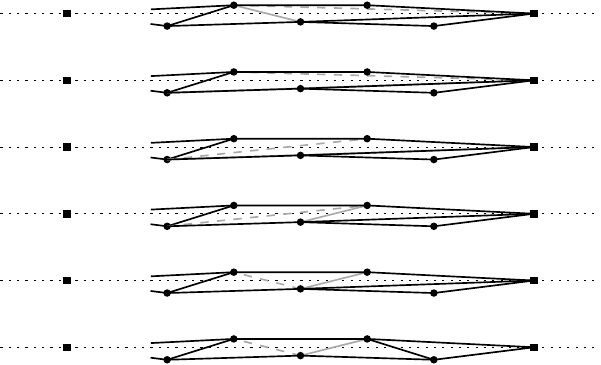}
\caption{Changing the labels of the edges $ct_1$ and $ct_2$ requires only a constant number of flips if one of $v$ or $w$ is ``visible''.
}
\label{fig:point_set_clause_2}
\end{figure}

\begin{lemma}
The edge center gadget of each edge $vw$ of $G$ can be placed in such a way that any flip sequence that swaps the labels of $ct_1$ and $ct_2$ contains a triangulation where one of the involved wiring gadgets has all edges from $L$ and $R$ incident to its defining vertex (in the way shown in \fig{fig_wiring},~right).
\end{lemma}
\begin{proof}
We explore the neighborhood of the edge $ct_2$ in the quadrilateral graph of $S$.
The edge center can be made small enough such that the only edges of $S$ that cross $ct_2$ are $t_1 b_2$, $b_1 w$, $b_2 v$, and $vw$.
We claim that, apart from $ct_2$, the edge $t_1 b_2$ has only neighbors in the quadrilateral graph that involve a point of the edge center and one of $v$ and~$w$.
This can be assured by placing the points of the gadget close enough to $e$; any quadrilateral spanned by $t_1$, $b_2$ and two other points of $S \setminus \{v,w\}$ must contain at least one of $c$ and $t_2$ and therefore is not an edge in the quadrilateral graph.
Hence, in order to swap the labels of $ct_1$ and $ct_2$, there has to be an edge between one of $v$ and $w$ and a point of the edge center gadget.
This, in turn, means that each of the $2w-2$ edges of the corresponding wiring need to be flipped.
Indeed, when flipping these edges to be incident to, say, $w$, the edges $ct_1$ and $ct_2$ become the diagonals of an empty convex pentagon $(w, t_2, t_1, b_1, c)$, and their labels can be swapped (see \fig{fig:point_set_clause_2}).
\end{proof}

Following~\cite{point_set_hard}, we choose $w$ large enough such that the number of all other flips that are required to change the labeling in each edge center are less than $\alpha w$, for some constant $0 < \alpha \ll 1$.

For the shortest flip sequence, an algorithm implicitly needs to determine the minimal number of wirings to flip to their center, and this corresponds to a minimum vertex cover of the original graph~$G$.
By dividing the length of a short flip sequence by $4w-2$ and rounding down, we thus get the size of a small vertex cover.
The computations in~\cite{point_set_hard} show that an approximation for the vertex cover can be obtained from an according approximation of the flip sequence, and thus show APX-hardness.

\begin{theorem}
Given an edge-labeled triangulation of a point set together with a permutation of its labels, it is APX-hard to find the shortest sequence of flips to obtain the given permutation of the labels.
\end{theorem}

\section{Simple Polygons}
The reduction for the flip distance problem on simple polygons in~\cite{poly_hard} is from the strongly NP-complete \textsc{Rectilinear Steiner Arborescence} problem~\cite{shi_su}.
In this problem, we are given a set of points (called \emph{sinks}) with positive integer coordinates and an integer $K$, and determine whether there is a directed tree rooted at the origin, consisting of horizontal and vertical edges oriented from left to right, or from bottom to top that contains all sinks, having total length at most~$k$.
An instance of the problem is mapped to a polygon in which we choose a certain initial triangulation.
The triangulation of a polygon $Q$ formed by a large double chain with an additional extreme point $p$ inside the flip-kernel is mapped to a monotone path on the grid, starting at the origin.
Flips in the triangulation modify that path.
In particular, flipping an edge such that it becomes incident to $p$ shortens the path, whereas flipping an edge incident to $p$ extends it.
See \fig{fig_chain_path}.
The polygon $Q$ is augmented by \emph{sink gadgets}.
A sink $(i,j)$ of a \textsc{Rectilinear Steiner Arborescence} problem is mapped to a ``small'' double chain that is between the points points $l_j$ and $l_{j+1}$ of the lower chain of the large double chain of~$Q$, and has only $u_i$ in its flip-kernel.
See \fig{fig_installing_sites}.
Every triangulation of the whole polygon can be mapped to a \emph{local triangulation} of~$Q$: an edge with an endpoint on a chain $C$ of a sink gadget is mapped to an edge with that endpoint being the first vertex of~$C$, i.e., the vertex also present in $Q$.
The proof relies on the fact that for a short flip sequence, the points of such a small \emph{sink gadget} need to see the only point in their flip-kernel.
Via the mapping to the local triangulation of the large double chain, this corresponds to the path passing over a sink in the original problem, and the union of all paths gives a short arborescence.

\begin{figure}
\centering
\includegraphics{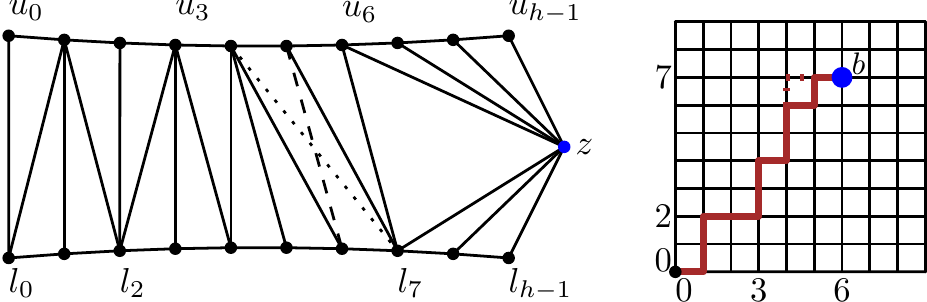}
\caption{The mapping of a triangulation of the polygon~$Q$ (consisting of a double chain with one additional point in the flip-kernel) to a path from the origin in the grid.
The paths correspond to parts of a rectilinear Steiner arborescence.
Having a certain triangle in corresponds to the path (and thus the Steiner arborescence) to visit a vertex.
Figure from~\cite{poly_hard}.
}
\label{fig_chain_path}
\end{figure}

\begin{figure}
\centering
\includegraphics[width=\textwidth]{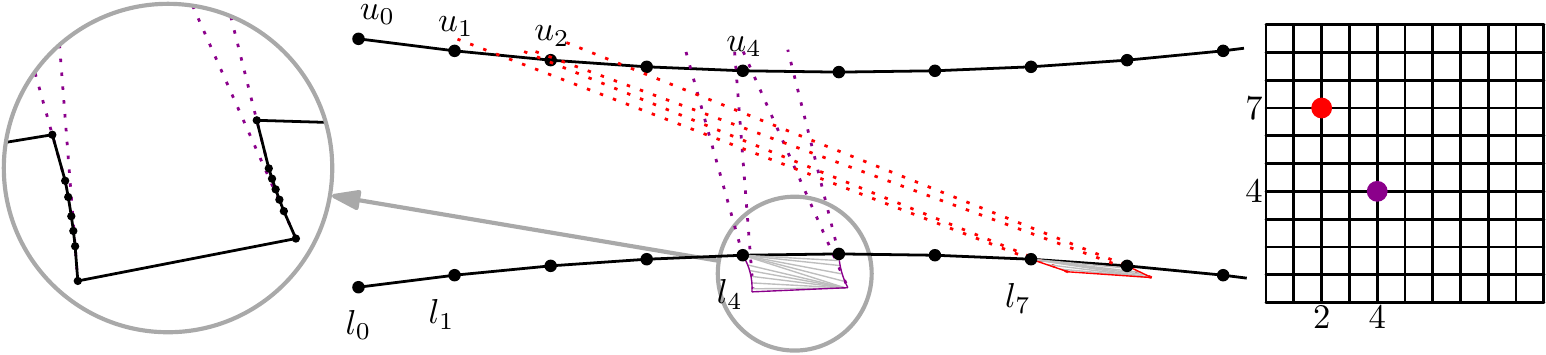}
\caption{Placement of small double chains as sink gadgets in the original reduction.
Figure from~\cite{poly_hard}.}
\label{fig_installing_sites}
\end{figure}

Our current reduction is almost the same.
However, the sink gadgets of the previous reduction are replaced by the ones shown in \fig{fig:new_sink}.
(They actually behave similar to the edge center gadgets in the reduction for point sets.)
For a sink gadget that is at the edge $l_i l_{i+1}$ and is ``aimed'' at vertex $u_j$ (i.e., representing a sink with coordinates $(i, j)$), we add three additional vertices $a$, $b$, and $c$.
The vertices $a$ and $b$ are placed such that they only see the vertex~$u_j$.
(This can always be achieved by placing $c$ close to the edge $a l_i$.)
We require that the labels of $ca$ and $c l_i$ are swapped in the target labeling.
The only quadrilaterals inside the polygon having the edge $ca$ as a diagonal are with the diagonals $bu_j$ and $b l_{i}$.
The only further quadrilateral inside the polygon having the edge $b l_{i}$ as a diagonal has the diagonal $a u_j$.
Hence, for changing the label of the edge $ac$, one of the edges $a u_j$ or $b u_j$ must be present in some triangulation of the flip sequence.

Consider now this triangulation and the mapping of it to a local triangulation of~$Q$.
An edge $va$ from any vertex $v$ to $a$ is mapped to the edge $vl_i$, and an edge $vb$ or $wc$ is mapped to an edge $vl_{i+1}$ or $wl_{i+1}$, respectively.
Observe that this mapping does not introduce crossings (a detailed analogous argument can be found in~\cite{poly_hard}), and therefore it results in a triangulation of~$Q$, the \emph{local triangulation}.
If an edge from $u_j$ to $a$ or $b$ is present, then this local triangulation contains the triangle $u_j l_i l_{i+1}$.
This means that the path corresponding to the local triangulation contains the sink~$(i, j)$, as required for the reduction.
The remaining argument is analogous to the one in~\cite{poly_hard}.

\begin{figure}
\centering
\includegraphics[page=2]{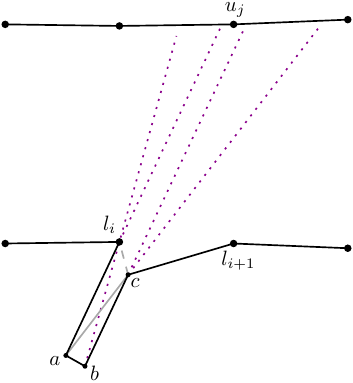}
\caption{A sink gadget for the labeled case, representing a sink $(i, j)$.
To swap the labels of the edges $ca$ and $cl_i$, they need to ``see'' the vertex $u_j$, i.e., there is a triangulation in the flip sequence that contains an edge from $u_j$ to a vertex of the gadget.
This results in a local triangulation of the polygon~$Q$ that contains the triangle~$u_j l_i l_{i+1}$.}
\label{fig:new_sink}
\end{figure}

\begin{theorem}
Given an integer $k$, an edge-labeled triangulation of a simple polygon and a permutation of its labels, it is NP-complete to decide whether there exists a sequence of at most $k$ flips to obtain this permutation of the labels.
\end{theorem}

\section{Conclusion}
We showed that previous hardness reductions for the flip distance problem can be extended for the setting of edge-labeled triangulations.
In particular, the problems remain hard if the source and target triangulations are the same except for the labels of certain edges.

Since flipping from the upper extreme triangulation to the lower extreme triangulation using $4n-4$ flips gives a permutation of the edge labels that is independent of which point in the flip-kernel is chosen to flip the edges to, determining a short flip sequence between two triangulations is hard even if we are given a labeling that results from such a short sequence.

The modifications presented in this note merely allow for having the (unlabeled) source and target triangulations identical.

The complexity of the flip distance problem for convex polygons remains unsolved.


\bibliographystyle{abbrv} %
\bibliography{bibliography}

\begin{thebibliography}{10}

\bibitem{empty5gon}
Z.~Abel, B.~Ballinger, P.~Bose, S.~Collette, V.~Dujmovi{}\'c, F.~Hurtado,
  S.~Kominers, S.~Langerman, A.~P{\'o}r, and D.~Wood.
\newblock Every large point set contains many collinear points or an empty
  pentagon.
\newblock {\em Graphs Combin.}, 27:47--60, 2011.

\bibitem{poly_hard}
O.~Aichholzer, W.~Mulzer, and A.~Pilz.
\newblock Flip distance between triangulations of a simple polygon is
  {NP}-complete.
\newblock {\em Discrete Comput. Geom.}, 54(2):368--389, 2015.

\bibitem{edge_labelled}
P.~Bose, A.~Lubiw, V.~Pathak, and S.~Verdonschot.
\newblock Flipping edge-labelled triangulations.
\newblock {\em Comput. Geom.}, 68:309--326, 2018.

\bibitem{eppstein}
D.~Eppstein.
\newblock Happy endings for flip graphs.
\newblock {\em JoCG}, 1(1):3--28, 2010.

\bibitem{harborth}
H.~Harborth.
\newblock {Konvexe F{\"u}nfecke in ebenen Punktmengen}.
\newblock {\em {Elemente der Mathematik}}, 33:116--118, 1978.
\newblock In German.

\bibitem{hurtado_noy_urrutia}
F.~Hurtado, M.~Noy, and J.~Urrutia.
\newblock Flipping edges in triangulations.
\newblock {\em Discrete Comput. Geom.}, 22:333--346, 1999.

\bibitem{lawson_connected}
C.~L. Lawson.
\newblock Transforming triangulations.
\newblock {\em Discrete Math.}, 3(4):365--372, 1972.

\bibitem{lawson_delaunay}
C.~L. Lawson.
\newblock Software for {$C^1$} surface interpolation.
\newblock In J.~R. Rice, editor, {\em Mathematical Software III}, pages
  161--194. Academic Press, NY, 1977.

\bibitem{orbit_conjecture_socg}
A.~Lubiw, Z.~Mas{\'{a}}rov{\'{a}}, and U.~Wagner.
\newblock A proof of the orbit conjecture for flipping edge-labelled
  triangulations.
\newblock In B.~Aronov and M.~J. Katz, editors, {\em 33rd International
  Symposium on Computational Geometry, SoCG 2017, July 4-7, 2017, Brisbane,
  Australia}, volume~77 of {\em LIPIcs}, pages 49:1--49:15. Schloss Dagstuhl -
  Leibniz-Zentrum fuer Informatik, 2017.

\bibitem{lubiw_pathak}
A.~Lubiw and V.~Pathak.
\newblock Flip distance between two triangulations of a point set is
  np-complete.
\newblock {\em Comput. Geom.}, 49:17--23, 2015.

\bibitem{vertex_cover_apx}
C.~H. Papadimitriou and M.~Yannakakis.
\newblock Optimization, approximation, and complexity classes.
\newblock {\em J. Comput. Syst. Sci.}, 43(3):425--440, 1991.

\bibitem{point_set_hard}
A.~Pilz.
\newblock Flip distance between triangulations of a planar point set is
  {APX}-hard.
\newblock {\em Comput. Geom.}, 47(5):589--604, 2014.

\bibitem{pilz_thesis}
A.~Pilz.
\newblock {\em On the complexity of problems on order types and geometric
  graphs}.
\newblock PhD thesis, 2014.

\bibitem{shi_su}
W.~Shi and C.~Su.
\newblock The rectilinear {S}teiner arborescence problem is {NP}-complete.
\newblock In {\em Proc. 11\textsuperscript{th} SODA}, pages 780--787, 2000.

\bibitem{problemas}
J.~Urrutia.
\newblock Algunos problemas abiertos.
\newblock In {\em Proc. IX Encuentros de Geometr{\'i}a Computacional}, pages
  13--24, 2001.
\newblock In Spanish.

\end{thebibliography}

\end{document}